\documentclass[twocolumn,english,ruled]{IEEEtran}
\usepackage[latin9]{inputenc}
\usepackage{babel}
\usepackage{booktabs}
\usepackage{algorithm2e}
\usepackage{amsmath}
\usepackage{amsthm}
\usepackage{stackrel}
\usepackage{graphicx}
\usepackage[unicode=true,
 bookmarks=true,bookmarksnumbered=true,bookmarksopen=true,bookmarksopenlevel=1,
 breaklinks=false,pdfborder={0 0 0},pdfborderstyle={},backref=false,colorlinks=false]
 {hyperref}
\hypersetup{pdftitle={Your Title},
 pdfauthor={Your Name},
 pdfborderstyle={},pdfborderstyle={},pdfborderstyle={},pdfborderstyle={},pdfborderstyle={},pdfborderstyle={},pdfborderstyle={},pdfpagelayout=OneColumn,pdfnewwindow=true,pdfstartview=XYZ,plainpages=false}
\usepackage{breakurl}

\makeatletter

\providecommand{\tabularnewline}{\\}

\theoremstyle{plain}
\newtheorem{thm}{\protect\theoremname}
\theoremstyle{definition}
\newtheorem{defn}[thm]{\protect\definitionname}
\theoremstyle{plain}
\newtheorem{lem}[thm]{\protect\lemmaname}
\theoremstyle{plain}
\newtheorem{fact}[thm]{\protect\factname}

\usepackage{babel}
\usepackage{babel}
\usepackage{babel}
\usepackage{babel}
\usepackage{babel}
\usepackage{babel}
\usepackage{babel}
\usepackage{babel}
\ifCLASSOPTIONcompsoc
\else
\fi

\providecommand{\definitionname}{defn}
\providecommand{\theoremname}{thm}

\providecommand{\definitionname}{defn}
\providecommand{\lemmaname}{lem}
\providecommand{\theoremname}{thm}

\providecommand{\definitionname}{Definition}
\providecommand{\lemmaname}{Lemma}
\providecommand{\theoremname}{Theorem}

\providecommand{\definitionname}{Definition}
\providecommand{\lemmaname}{Lemma}
\providecommand{\theoremname}{Theorem}

\providecommand{\definitionname}{Definition}
\providecommand{\lemmaname}{Lemma}
\providecommand{\theoremname}{Theorem}

\providecommand{\definitionname}{Definition}
\providecommand{\lemmaname}{Lemma}
\providecommand{\theoremname}{Theorem}

\providecommand{\definitionname}{Definition}
\providecommand{\lemmaname}{Lemma}
\providecommand{\theoremname}{Theorem}

\providecommand{\definitionname}{Definition}
\providecommand{\factname}{Fact}
\providecommand{\lemmaname}{Lemma}
\providecommand{\theoremname}{Theorem}

\makeatother

\providecommand{\definitionname}{Definition}
\providecommand{\factname}{Fact}
\providecommand{\lemmaname}{Lemma}
\providecommand{\theoremname}{Theorem}

\begin{document}

\title{Point Sweep Coverage on Path}

\author{Dieyan Liang\thanks{School of data science and computer science, Sun Yat-sen University,
Guangzhou China},~ Hong Shen\thanks{1. School of data science and computer science, Sun Yat-sen University,
Guangzhou China 2. School of Computer Science, University of Adelaide,
Australia},~}
\maketitle
\begin{abstract}
An important application of wireless sensor networks is the deployment
of mobile sensors to periodically monitor (cover) a set of points
of interest (PoIs). The problem of Point Sweep Coverage is to deploy
fewest sensors to periodically cover the set of PoIs. For PoIs in
a Eulerian graph, this problem is known NP-Hard even if all sensors
are with uniform velocity. In this paper, we study the problem when
PoIs are on a line and prove that the decision version of the problem
is NP-Complete if the sensors are with different velocities. We first
formulate the problem of Max-PoI sweep coverage on path (MPSCP) to
find the maximum number of PoIs covered by a given set of sensors,
and then show it is NP-Hard. We also extend it to the weighted case,
Max-Weight sweep coverage on path (MWSCP) problem to maximum the sum
of the weight of PoIs covered. For sensors with uniform velocity,
we give a polynomial-time optimal solution to MWSCP. For sensors with
constant kinds of velocities, we present a $\frac{1}{2}$-approximation
algorithm. For the general case of arbitrary velocities, we propose
two algorithms. One is a $\frac{1}{2\alpha}$-approximation algorithm
family scheme, where integer $\alpha\ge2$ is the tradeoff factor
to balance the time complexity and approximation ratio. The other
is a $\frac{1}{2}(1-1/e)$-approximation algorithm by randomized analysis. 
\end{abstract}


\section{Introduction}

Coverage is one of the most important applications of wireless sensor
networks (WSNs). In applications of WSN, sensors are placed in an
area of interest (AoI) to monitor the environment and detect extraordinary.
Coverage problem have been gotten much attention because of its importance.
Many researches have been studied on this topic on various scenes,
including discrete points\cite{gorain2013point}, 2-dimensional surface,
3-dimensional surface\cite{kong2014surface}, 3-dimensional space,
fence\cite{li2015minimizing} and so on. Based on sensors' characteristic,
special factors should be considered such as energy efficiency, maintaining
connectivity and so on\cite{yu2014connected,sengupta2013multi}.

However, most of the existing work mainly focus on continuous coverage,
where sensors stay still after placed on the objective locations.
Not until recently did sweep coverage be brought up to study the periodical
coverage situation. Periodical coverage is also a typical application.
For example, guards need to patrol barrier periodically where intruders
need some time to cross; Information collector should collect information
from the objective sensors periodically in case that their memory
overflow. In those situations, the objects do not need to covered
continuously. So mobile sensors can move around to cover more objects
than continuous coverage to reduce monitoring cost. Therefore, sweep
coverage begins to be payed attention after being brought up\cite{cheng2008sweep,collins2013optimal,gorain2013point,gorain2014line,pasqualetti2012cooperative}.

In this paper, we focus on the point sweep coverage problem in which
a set of discrete points need to be covered by a given set of mobile
sensors at least once within a given period of time. The point sweep
coverage problem was first brought up in the paper\cite{cheng2008sweep}
that showed finding minimum number of mobile sensors with a constant
velocity to cover PoIs in an Eulerian graph is NP-Hard and can not
be approximated within a factor of 2. Until now, there have not been
papers to discuss the situation in which PoIs are sweep covered by
mobile sensors with arbitrary velocities because of its hardness.
Nevertheless, it is an important situation because the maximum velocities
of mobile sensor would be reduced along with their energy used up.
In this paper, we start to study the situation when mobile sensors
have arbitrary velocities. It must be a NP-Hard problem when PoIs
in graph. We want to study a simpler scene that PoIs are distributed
on path. We call it point sweep coverage on path (PSCP) problem, which
is to find whether or not a given set of mobile sensors cover all
the PoIs on path periodically. 

Point sweep coverage problem on path is a simpler scene than graph
but also of good practical use. For example, as illustrated in Figure
1, numbers of static sensors are placed on key locations to collect
real-time information of ocean or resources. A set of mobile sensors
are given to collect the data from the static sensors periodically
in case that the memory of the static sensors overflow. It is one
of classic practical scene of PSCP problem. Besides, PSCP problem
is of practical use in security, forest conservation, resource exploration
and so on. It makes PSCP problem necessary to study.

\begin{figure}
\caption{Mobile sensors sweep coverage around island to collect the data from
the static sensors}

\centering\includegraphics[scale=0.5]{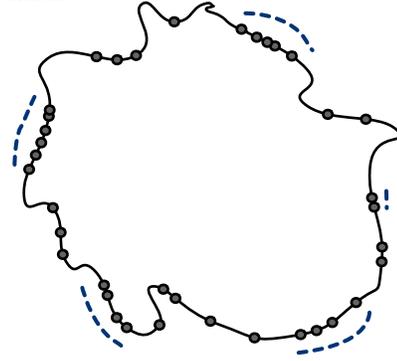}
\end{figure}

% \begin{figure}
%  \includegraphics{./ocean bay}
%  \caption{PoIs on bay to be sweep covered}
%  \end{figure}

In this paper, we prove that PSCP problem is NP-Complete and define
its optimization problem, Max-Weight sweep coverage on path (MWSCP),
which is NP-Hard. The MWSCP problem is to find the maximum sum of
weight of PoIs covered by a given set of mobile sensors when PoIs
are distributed on a path. When the weight of PoIs is one, MWSCP problem
turns to be a special case, Max-PoI sweep coverage on path (MPSCP)
problem, whose object is to maximum the number of PoIs sweep covered.
We analyze the special cases and general case of MWSCP respectively,
propose an optimal algorithm for the uniform velocity case and approximation
algorithms for others.

The main contributions of this paper are summarized as follows: 
\begin{itemize}
\item We define the point sweep coverage on path (PSCP) and its variant
problems, Max-Weight sweep coverage on path (MWSCP) problem. 
\item We prove PSCP is NP-Complete and MWSCP is NP-Hard. 
\item For the special cases of MWSCP problem when sensors have uniform velocity,
we present an optimal algorithm. 
\item For the special cases of MWSCP problem when sensors have constant
number of velocities, we propose a $\text{\ensuremath{\frac{1}{2}}}$-approximation
algorithm. 
\item For the general cases of MWSCP problem when sensors have arbitrary
velocities, we propose two algorithms. One is a $\frac{1}{2\alpha}$-approximation
algorithm family scheme, where integer $\alpha\ge2$ is the tradeoff
factor to balance the time complexity and approximation ratio. And
the other is a $\frac{1}{2}(1-1/e)$-approximation algorithm by randomized
analysis. 
\end{itemize}
The rest of this paper is organized as follows: Section 2 describes
some related work. In Section 3, the definition and NP-Completeness
proof of the PSCP problem are given. In Section 4, we define MPSCP
problem and MWSCP problem, present our optimal and approximation algorithms
for different cases of MWSCP problem. Section 5 concludes the paper.

\section{Related work}

Point sweep coverage was firstly brought up in \cite{cheng2008sweep}.The
authors presented a min-sensor point sweep coverage problem (MSPSC)
to find the minimum number of sensors to sweep cover PoIs in Eulerian
graph, which was proved NP-Hard by transforming TSP problem to it.
The problem could not be approximated within ratio 2 and local algorithms
can not work. In \cite{du2010sweep}, the authors distinguished sensors'
strategies, proposed MinExpand algorithm for un-cooperated sensors
and Osweep algorithm for cooperated sensors respectively. The mistake
of approximation analysis of prior papers was rectified by Gorain
et al.\cite{gorain2015approximation}. 3-approximation algorithm for
MSPSC problem was proposed and it may be the best approximation algorithm
for MSPSC until now. No polynomial time constant factor approximation
algorithm for MSPSC was also proved. Some variants of MSPSC were also
presented. When PoIs had different sweep periods, a $\mathcal{\text{O(log}\rho})$-approximation
algorithm were proposed, where $\rho$ was the ratio of the maximum
and minimum sweep periods among PoIs. And there were area sweep coverage
problem and line sweep coverage problem proposed\cite{gorain2014approximation,gorain2014line},
The problem of area sweep coverage are NP-Hard and a $(\sqrt{2}+\frac{2-\sqrt{2}}{mn})$-approximation
algorithm for that were proposed. The line sweep coverage problem
was also NP-Hard. Gorain et al. proposed a 2-approximation algorithm
and proved the problem can not be approximated within 2. In \cite{chen2016efficient},
the effect of sensing range was considered. The authors proposed DSRS
problem and proved it NP-Hard. In \cite{gorain2016solving}, energy
consumption was taken into consideration. Two new problem were proposed.
One was to minimizing energy consumption and the other was to minimizing
the number of sensors when energy consumption was bounded. All the
above papers only focused on the point sweep coverage problem when
sensors had uniform velocity. No approximation algorithm had been
brought up for sensors having different velocities yet.

Even though the point sweep coverage problem was presented firstly
by \cite{cheng2008sweep}, the concept of sweep coverage initially
came from the context of robotics. The researches on robotics often
focused on sweep covering continuous lines, and the problem was called
boundary patrolling or fence patrolling. In these researches, the
aim was to find the minimized idleness, i.e., the longest time interval
during which there is at least one point on the boundary still uncovered
by any mobile sensors. In \cite{czyzowicz2011boundary}, the authors
firstly studied boundary patrolling problem and proposed two intuitional
algorithms for open and close fence patrolling. The optimality of
the intuitional algorithms was disproved in \cite{dumitrescu2014fence,kawamura2015fence,kawamura2015simple}.
Some examples were proposed to illustrate that the idleness could
be reduced to 41/42, 24/25 even 3/4 by special design, assuming the
idleness of proportional solution presented in \cite{czyzowicz2011boundary}
is 1 . Even though the optimality of algorithms presented by Czyzowicz
et al. was disproved, the optimal solution had not been brought up
yet. In \cite{pasqualetti2012cooperative}, the authors extended the
scenes of the Min-idleness point sweep coverage problem to chain,
tree, and cyclic roadmap. Within tolerance $\epsilon$, they could
get an optimal idleness when PoIs were on chain in time complexity
$O(nlog(\epsilon^{-1}))$. And a 8-approximation algorithm was proposed
to find min-idleness when PoIs were on cyclic roadmap. The problem
is NP-Hard. In \cite{collins2013optimal}, the authors described a
fragmented boundaries environment and found the optimal solution for
Min-idleness in that environment.

\section{Point sweep coverage on path}

In this section, we will study PSCP problem. Firstly we present the
definition of PSCP problem. Then we prove PSCP problem NP-complete.

The definition of point sweep coverage on path is given below according
to the definition of point sweep coverage\cite{cheng2008sweep}. 
\begin{defn}
(Point Sweep Coverage on Path problem) Given a set of PoIs $P=\{p_{1},p_{2},..,p_{N}\}$
distributed on a path, each one $p_{i}$ needs to be covered every
$T$ time period. If $p_{i}$ is covered within $T$ time period,
then we call $p_{i}$ is $T$-sweep covered. Time period $T$ is also
called sweep period. Given $M$ mobile sensors $S=\{s_{1},s_{2},..,s_{M}\}$
patrolling the $N$ PoIs, find whether the mobile sensors can sweep
cover all PoIs. 
\end{defn}
Without loss of generality, when PoIs are on an open path, we can
simplify that the set of PoIs are located on x-line with x-coordination
$X=\{x_{1}=0,x_{2},...,x_{N}\}(x_{1}<x_{2}<...<x_{N})$. Let $V=\{v_{1},v_{2},...,v_{M}\}$
where $v_{i}$ is the velocity of sensor $s_{i}\in S$ $(1\leq i\leq M)$.
PoIs have uniform sweep period $T$. Let $L$ denotes the line segment
from $1^{st}$ PoI to $N^{th}$ PoI. The algorithm for open path can
be easily transformed to the one for close path, so we do not belabor
the algorithm for close path.

From the existing work, we know there are thousands of strategies
for mobile sensors to cover PoIs and mainly be classified to 2 kinds
. Between those, there is a kind of simple strategy, separated strategy,
by which mobile sensors move back and forth to cover a line segment
without cooperating with others and each PoI is only covered by the
same mobile sensor periodically. 
\begin{defn}
(Separated strategy)\cite{collins2013optimal}. Each mobile sensor
moves back and forth on the line segment, and the trajectories of
different sensor do not overlap with each other. 
\end{defn}
Separated strategy is not necessarily the optimal strategy. There
are complicated strategies by which some sensors cooperate with others
to cover the same line segment and PoIs are sweep covered by more
than one mobile sensor, which are classified as cooperated strategy.
Some examples are proposed to show that separated strategy may slightly
worse than some complicated strategy \cite{dumitrescu2014fence,kawamura2015fence}
in some special cases, i.e., the covering range covered by a set of
mobile sensors using separated strategy would be slightly shorter
than some complicated strategy. However, until now, how to cooperate
among sensors to get the optimal monitoring efficiency is still unknown.

Using separated strategy, every mobile sensor has its own covering
region. The range of the region is denoted by $r_{i}$, $r_{i}\leq v_{i}T/2$
for each $s_{i}\in S$. In optimal deployment, we can just assume
$r_{i}=v_{i}T/2$. Then a sensor's trajectory can be denoted by the
first PoI's location in its covering region and its covering range.
Since the covering range of a sensor has already decided by the velocity
of the sensor in optimal separated strategy, we can just use the first
PoI's locations in sensors' covering region to denote the separated
deployment.

The notations will illustrated in table 1.

\begin{table*}
\begin{centering}
\caption{Notations}
\par\end{centering}
\centering{}%  \label{tab:commands}
\begin{tabular}{ccl}
\toprule 
Symbol  & Definition  & \tabularnewline
\midrule 
$N$  & the number of the PoIs  & \tabularnewline
$M$  & the number of the sensors  & \tabularnewline
$P$  & the set of PoIs $\{p_{1},p_{2},..,p_{N}\}$  & \tabularnewline
$S$  & the set of mobile sensors $\{s_{1},s_{2},..,s_{M}\}$  & \tabularnewline
$X$  & the locations of PoIs $\{x_{1},x_{2},\ldots,x_{N}\}$ $(x_{1}<x_{2}<...<x_{N})$  & \tabularnewline
$T$  & the uniform sweep period of PoIs  & \tabularnewline
$V$  & velocities of the sensors $\{v_{1},v_{2},...,v_{M}\}$  & \tabularnewline
$W$  & the weight of the PoIs $\{w_{1},w_{2},...,w_{N}\}$  & \tabularnewline
\bottomrule
\end{tabular}
\end{table*}

\subsection{Problem Hardness}

In this section, we want to prove PSCP problem NP-Complete. Before
the proof, let us recall the definition of the 3-Partition problem
and quote a theorem from paper\cite{kawamura2015fence}. Then in the
proof, we will transform 3-Partition problem to PSCP problem. 
\begin{defn}
(3-Partition)\cite{gary1979computers}

INSTANCE: Set A of 3m elements, a bound $B\in Z^{+}$, and a size
$s(a)\in Z^{+}$for each $a\in A$ such that $B/4<s(a)<B/2$ and such
that $\sum_{a\in A}s(a)=mB.$

QUESTION: Can A be partitioned into m disjoint sets $A_{1},A_{2},...,A_{m}$
such that, for $1\leq i\leq m$, $\sum_{a\in A_{i}}s(a)=B$ (note
that each $A_{i}$ must therefore contain exactly three elements from
A) ? 
\end{defn}
\begin{thm}
\cite{kawamura2015fence}For three sensors or two sensors, the separated
strategy is optimal. 
\end{thm}
\begin{thm}
Point sweep coverage on path is NP-Complete problem. 
\end{thm}
\begin{proof} Given a 3-Partition instance like definition 3, we
construct a PSCP instance. Given a set of $N$ PoIs, where $N=(2B+2)*m$,
the positions of PoIs are $\{x_{1},x_{2},...,x_{N}\}(x_{1}<x_{2}<...<x_{N})$.
$d_{i}=x_{i+1}-x_{i}(1\leq i<N)$ means the distance between $(j+1)^{th}$
PoI and $j^{th}$ PoI. The positions satisfy the equation below.

\begin{alignat}{1}
x_{1} & =0\\
d_{j} & =B\nonumber \\
 & (j=k*(2B+2);1\leq k\leq m-1)\\
d_{j} & =B/(2B+1)\nonumber \\
 & ((k-1)*(2B+2)<j\leq k*(2B+2);1\leq k\leq m)
\end{alignat}

Given a set $S$ of $M=3m$ mobile sensors, the velocity of mobile
sensor $s_{i}$ is $v_{i}=2*s(a_{i})/T$ for $a_{i}\in A$ $(1\leq i\leq M)$.
As mentioned before, if using separated strategy, each mobile sensor
has their own covering range $r_{i}=v_{i}T/2=s(a_{i})$, $1\leq i\leq M$.
So we get $B/4<r_{i}<B/2$, $\sum_{1\leq i\leq M}r_{i}=mB.$

Let $B_{j}$ denote the line segment from $x_{(j-1)*(2B+2)+1}$ to
$x_{j*(2B+2)}$ for $1\leq j\leq m$. Because of equation (3), we
get $|B_{j}|=B$. If the 3-Partition instance is satisfied, we get
$\sum_{a\in A_{i}}s(a)=B$, where each $A_{i}$ contain exactly three
elements from $A$. It means we can obtained proper 3 mobile sensors
to sweep cover each $B_{j}$ for $1\leq j\leq m$., then we get a
proper deployment for PSCP problem.

Conversely, for the other side, if we have an unique solution for
PSCP problem, because of equation (2), mobile sensors $s_{i}$ $(1\leq i\leq M)$
must cover some part of line segment $B_{j}$ $(1\leq j\leq m)$ since
the gap between $B_{j}$ and the other is too big. According to theorem
4, because of $B/4<r_{i}<B/2$ for $1\leq j\leq m$, more than 2 sensors
would be needed to cover each segment $B_{j}$ $(1\leq j\leq m)$
. Considering there are 3m sensors for m line segment $B_{j}(1\leq j\leq m)$,
we get exactly 3 mobile sensors to cover each segment $B_{j}$ $(1\leq j\leq m)$
. According to theorem 4, the strategy is separated strategy. For
$B_{j}$ $(1\leq j\leq m)$, assuming covering ranges of the 3 mobile
sensor are $r_{j1},r_{j2},r_{j3}$ respectively, we get 
\[
r_{j1}+r_{j2}+r_{j3}\geq B-2B/(2B+1)
\]

And because $r_{j1},r_{j2},r_{j3}\in Z^{+}$, we get $r_{j1}+r_{j2}+r_{j3}\geq B$.
If $r_{j1}+r_{j2}+r_{j3}>B$, there must exists $r_{i1}+r_{i2}+r_{i3}<B(i\neq j)$,
making it contradiction. So $r_{j1}+r_{j2}+r_{j3}=B$ for $1\leq j\leq m$
. 3-Partition problem is satisfied.

It is clear that the instance of PSCP can be constructed from arbitrary
3-Partition instance in polynomial time of $N$. An example is illustrated
in figure 2.

\begin{figure}
\caption{The instance of PSCP when $m=3$, $B=7$. It contains 48 PoIs and
9 sensors. The PoIs are distributed on subsegment $B_{j}$ $(1\le j\le3)$.
When the covering range of sensor $s_{i}$ is $\frac{7}{4}\le r_{i}\le\frac{7}{2}$
$(1\le i\le9)$, the proper deployment is to deploy 3 sensors to each
$B_{j}$.}

\includegraphics[scale=0.5]{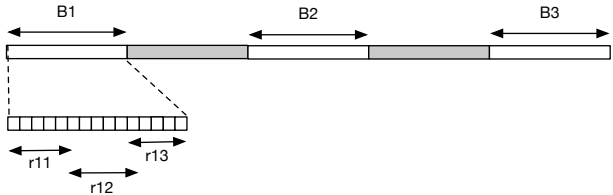} 
\end{figure}

It is easy to see that $PSCP\in NP$, since it can be checked in polynomial
time of $N$ whether the given set of mobile sensors is enough to
sweep cover the PoIs or not.

The theorem is proved. \end{proof}

In this process of the above proof, we would find that, considering
cooperated strategy or not, PSCP is NP-complete problem.

\section{Max-Weight sweep coverage on path problem }

We promote an optimization problem for PSCP problem, called Max-PoI
Sweep Coverage on Path (MPSCP) problem, which is a special case of
Max-Weight sweep coverage on path problem.
\begin{defn}
(Max-PoI Sweep Coverage on Path Problem) The Max-PoI Sweep Coverage
on Path problem is to find the maximum number of PoIs sweep covered
on path by the given set of mobile sensors. 
\end{defn}
In the existing work, there are 2 kinds of optimization problems on
point sweep coverage problem proposed. One is Min-sensor Sweep Coverage
problem to find minimum number of mobile sensors to sweep cover all
the PoIs, the other is Min-idleness sweep coverage problem to minimize
maximum time period of the PoIs. However, in applications, resources
are always limited. We want to decide whether the given set of mobile
sensors can sweep cover the PoIs and if it is not enough, people would
like to cover as many PoIs as possible. Max-PoI Sweep Coverage (MPSC)
problem is to find the maximum number of PoIs sweep covered by the
given mobile sensors. It would maximize utilization efficiency of
the given mobile sensors. In theory, MPSC can be used as a decision
algorithm for point sweep coverage problem. When the output is $N$,
it means all PoIs can be sweep covered by the given set of mobile
sensors. We can also use the algorithm for MPSC as a subroutine to
fix the Min-sensor sweep coverage problem and the Min-idleness sweep
coverage. So it is essential to study MPSC problem.

In this paper, we focus on the point sweep coverage problem on path,
so MPSCP would be the our main content. Based on the NPC of PSCP problem,
MPSCP problem is a NP-Hard problem.

When we introduce PSCP problem before, we assume the PoIs are on different
x-coordinations. Actually, sometimes, some PoIs are on the same x-coordination;
or sometimes, PoIs would have different weight depending on their
locations' importance. MPSCP problem would be generalized to its weighted
case. We call the weighted case Max-Weight Sweep Coverage on Path
(MWSCP) Problem. The object is to find the maximum sum of the weight
of PoIs sweep covered, where each of PoIs would have a weight respectively
and the weight is allowed to be integer and fractional number. The
definition of MWSCP is below. 
\begin{defn}
(Max-Weight Sweep Coverage on Path Problem) Given a set of PoIs on
a path and a set of mobile sensors, each of which has a weight respectively,
the Max-Weight Sweep Coverage on Path problem is to find the maximum
sum of weight of PoIs sweep covered on path by the given set of mobile
sensors. 
\end{defn}
MWSCP is also NP-Hard problem. In the following text, we would fix
two kinds of special cases of MWSCP and its general case. We present
an optimal algorithm for MWSCP when mobile sensors have uniform velocity
and present a $\text{\ensuremath{\frac{1}{2}}}$-approximation algorithm
for the case when mobile sensors have constant kinds of velocities.
For the general case, we present a $\frac{1}{2\alpha}$-approximation
$(\alpha~is~an~integer,\alpha\ge2)$ algorithm family scheme using
rounding method and improve the approximation ratio to $\frac{1}{2}(1-1/e)$
using randomized algorithm.

If the number of PoIs is smaller than or equal to the number of mobile
sensors, the problem of MWSCP is trivial. In the following text, we
assume the number of PoIs is bigger than the number of mobile sensors,
i.e., $N>M$. Static sensors are regarded as special mobile sensors
with velocity 0.

\subsection{Max-Weight sweep coverage on path with uniform velocity}

In this subsection, we show there is a polynomial algorithm for optimal
solution for MWSCP when sensors have uniform velocity.

From the paper\cite{pasqualetti2012cooperative}, we know separated
strategy is optimal when PoIs on path are sweep covered by mobile
sensors with uniform velocity, because when two sensors meet while
moving in opposite directions they can ``exchange'' roles. 
\begin{lem}
Separated strategy is optimal strategy for point sweep coverage on
path problem if the given set of mobile sensors have uniform velocity. 
\end{lem}
Using separated strategy, every sensor has its own covering region.
From observation, we find MWSCP using separated strategy problem has
the nature of optimal substructure. So we use dynamic programming
to get the optimal solution. 
\begin{fact}
Dynamic programming would be an optimal algorithm to solve MWSCP using
separated strategy problem since it has the nature of optimal substructure. 
\end{fact}
In the case of uniform velocity, we set the uniform velocity of sensors
is $v$, then sensors have uniform covering range $r=vT/2$.

Let $OPT(i,j)$ denote the maximum sum of weight of PoIs covered by
$i$ sensors from $j^{th}$ PoIs, $n_{j}$ denote the biggest number
of PoIs covered by one sensor from $j^{th}$ PoI. $w_{j}$ is the
sum of weight of the $n_{j}$ PoIs, i.e., PoIs from $j^{th}$ to $(j+n_{j}-1)^{th}$.
The recursive formulation is given below:

\begin{multline*}
OPT(i,j)=\\
\begin{array}{c}
\begin{array}{cc}
max\left\{ \begin{array}{c}
OPT(i,j+1),\\
OPT(i-1,j+n_{j})+w_{_{j}}
\end{array}\right\}  & \begin{array}{c}
1\leq i\leq M,\\
1\leq j<N-1
\end{array}\end{array}\end{array}
\end{multline*}

And boundary conditions are

\begin{align*}
OPT(0,j) & =0 & 0\leq j\leq N\\
OPT(i,N) & =1 & 1\leq i\leq M
\end{align*}

The time complexity of the dynamic program is $O(M*N)$ . Since the
algorithm is straightforward so we omit it. The dynamic programming
algorithm is an optimal algorithm for MWSCP when sensors have uniform
velocity.

The PSCP problem when sensors have uniform velocity can also be judged
satisfied or not by making the sensors to cover from $1^{th}$ PoI
to $N^{th}$ continuously without overlapping.

\subsection{Max-Weight sweep coverage on path with constant number of velocities}

Now we discuss MWSCP problem when sensors have $K$ different velocities,
where $K$ is a constant. In fact, it is not hard to see that the
uniform velocity case of MWSCP problem is a special case of MWSCP
problem when sensors have constant kinds of velocities. The only difference
is that separated strategy is not an optimal strategy any more when
$K\ge2$ and $M\ge4$\cite{kawamura2015fence}. So we use dynamic
programming method to find an optimal separated strategy and propose
a $\frac{1}{2}$-approximation algorithm for MWSCP problem with constant
kinds of velocities considering the difference between separated strategy
and cooperated strategy.

Let $m_{i}$ be the number of sensors with velocity $v_{i}$, then
$M=\sum\nolimits _{i=1}^{K}m_{i}$. $OPT_{S}(i_{1},i_{2},...,i_{K},j)$
denotes the maximum sum of weight of PoIs covered when there are ${\sum\nolimits _{h=1}^{K}i_{h}}$
sensors to cover the line segment from $j^{th}$ PoI to $N^{th}$
PoI using separated strategy, where $i_{h}$ is the number of sensors
with velocity $v_{j}$ $(1\leq h\leq K)$. $n_{ij}$ denotes the maximum
number of PoIs covered by a sensor with velocity $v_{i}$ covering
from $j^{th}$ PoI and $w_{ij}$ is the sum of weight of the $n_{ij}$
PoIs. The recursive formulation lists below: for $1\leq i_{h}\leq M,1\leq h\leq K,1\leq j\leq N$,

\begin{multline*}
OPT_{S}(i_{1},i_{2},...,i_{K},j)=\\
max\left\{ \begin{array}{c}
OPT_{S}(i_{1},i_{2},...,i_{K},j+1),\\
OPT_{S}(i_{1}-1,i_{2},...,i_{K},j+n_{1j})+w_{1j},\\
OPT_{S}(i_{1},i_{2}-1,...,i_{K},j+n_{2j})+w_{2j},\\
...,\\
OPT_{S}(i_{1},i_{2},...,i_{K}-1,j+n_{Kj})+w_{kj}
\end{array}\right\} 
\end{multline*}

And boundary conditions are

\begin{alignat*}{1}
OPT_{S} & (0,0,...,0,j)=0\\
 & \qquad\qquad\qquad\hspace{1.5em}0\leq j\leq N\\
OPT_{S} & (i_{1},i_{2},...,i_{K},N)=1\\
 & 0\leq i_{h}\leq M,1\leq h\leq K,{\sum\nolimits _{h=1}^{K}i_{h}}\ge1
\end{alignat*}

In Algorithm 1, we maintain an array entry $Tr(j)$ for $j=1,...,K$.
Each entry $Tr(j)$ is a list of integers. A integer $i$ in the list
of entry $Tr(j)$ indicates that $i^{th}$ PoI is the first PoI in
the covering region of some sensor with velocity $v_{j}$, whose covering
range is $v_{j}T/2$.

% Algorithm
\begin{algorithm}[t]
\SetAlgoNoLine \KwIn{ The number of sensors $M$, the number of
PoIs $N$, sensors' velocities $V=\{v_{1},v_{2},...,v_{M}\}$, PoIs'
weight $W=\{\omega_{1},\omega_{2},...,\omega_{N}\}$, the number of
velocities $K$, sweep period $T$, PoIs' locations $X=\{x_{1},x_{2},...,x_{N}\}$,
an array entry $Tr(i)$ for $1\leq i\leq K$.} \KwOut{The maximum
sum of weight of PoIs covered: $O(m_{1},m_{2},...,m_{K},1)$} count
the numbers of sensors with different velocities respectively $\{m_{1},m_{2},...,m_{K}\}$\;

set $W[K][N]\leftarrow0$, $Tr[K]\leftarrow0$\;

\For{ $i\leftarrow1$ to $K$ }{ \For{ $j\leftarrow N$ to 1
}{ Let $N[i][j]$ be the biggest number of PoIs covered by one sensor
with velocity $v_{i}$ from $j^{th}$ PoI in sweep period $T$, i.e.,
$w_{ij}$\; Let $Wt[i][j]$ be the sum of weight of the $N[i][j]$
PoIs covered.\;}}

set $O[m_{1}+1][m_{2}+1]...[m_{K}+1][N]\leftarrow0$;\; initial $O[m_{1}+1][m_{2}+1]...[m_{K}+1][N]$
according to boundary conditions\;

\For{ $j\leftarrow N$ to 1}{

\For{ $i_{1}\leftarrow1$ to $m_{1}$ }{

\For{ $i_{2}\leftarrow1$ to $m_{2}$ }{

\For{...}{

\For {$i_{K}\leftarrow1$ to $m_{K}$ }{ call the recursive formulation
to calculate $O(i_{1},i_{2},...,i_{K},j)$\;

}}}}}

set $i_{1}=m_{1},i_{2}=m_{2},...,i_{K}=m_{K}$;\; \For{ $j\leftarrow1$
to N-1 }{

\If {$O(i_{1},i_{2},...,i_{K},j)!=O(i_{1},i_{2},...,i_{K},j+1)$
}{ \For{ $h\leftarrow1$ to $K$ }{ \If {$i_{h}>0$ $\&\&$
$O(i_{1},i_{2},...,i_{h},...,i_{K},j)==O(i_{1},i_{2},...,i_{h}-1,...,i_{K},j+W[i_{h}][j])+W[i_{h}][j]$
}{ add $j$ to Tr(h); $j-=W[i_{h}][j]$\; $i_{h}--$\; }}}}

\For{ $h\leftarrow1$ to $K$ }{

\If{ exist $i_{h}!=0$ }{

add $j$ to Tr(h); break\;

}}

\For{ $1\leq i\leq K$ }{ place sensors with velocity $v_{i}$
to the locations $Tr(i)$ respectively\; } \Return $O(m_{1},m_{2},...,m_{K},1)$;

\caption{MWSCP-K-Velocities}
\label{alg:one} 
\end{algorithm}

It will take $N*\prod_{i}m_{i}$ time to fix dynamic program in Algorithm
1 and take $N*K$ time to trace back to get the optimal deployment.
Note that, $N*\prod_{i}m_{i}\le N*(M/K)^{K}$, so Algorithm 3 will
take $O(N*(M/K)^{K})$ time. Because $K$ is a constant integer, the
algorithm is polynomial time algorithm. Now we want to prove that
the optimal algorithm for MWSCP using separated strategy is a $\frac{1}{2}$-approximation
algorithm for MWSCP. Before that, we prove a more generalized theorem. 
\begin{thm}
A $\beta$-approximation algorithm for MWSCP problem using separated
strategy can be turned to be a $\frac{\beta}{2}$-approximation for
MWSCP. 
\end{thm}
\begin{IEEEproof}
Let $A$ denote the $\beta$-approximation algorithm for MWSCP problem
using separated strategy, and $A_{O}$ denote the optimal algorithm.
In resulting deployment of algorithm $A_{O}$, there may be some line
segments covered by cooperated strategy and others covered by separated
strategy.

Case 1: For the line segments covered by separated strategy in Algorithm
$A_{O}$, the sum of weight of PoIs covered in these segments is $OPT_{1}$,
which is equal to the sum of weight of PoIs covered by optimal separated
strategy. If we sweep cover the same segment with its covering sensors
using algorithm $A$, the sum of weight of PoIs covered is $W_{1}\ge\beta*OPT_{1}$.

Case 2: For the line segments covered by cooperated strategy, we assume
the optimal sum of weight of PoIs covered is $OPT_{2}$. If we cover
the same segment with its covering sensors using optimal separated
strategy, we can cover $OPT_{S2}\ge OPT_{2}/2$ PoIs. For example,
$l\subseteq L$ is one of the line segments covered by a subset of
sensors $S'\subseteq S$ with velocities $\left\{ v_{1},v_{2},...,v_{n}\right\} $
$(n=|S'|)$ using cooperated strategy, in which the sum of weight
of PoIs covered is $W_{C}$, $|l|<\sum\nolimits _{i=1}^{n}{v_{i}}$.
If we cover $l$ by optimal separated strategy, the covering range
of sensors is $|{r}(S')|=\sum\nolimits _{i=1}^{n}{{v_{i}}/2}$ and
the sum of weight of PoIs covered is $W_{S}$ . If $W_{S}<W_{C}/2$,
the uncovered part of segment $l$ whose length is less than $|l|/2$,
which can also be covered by $S^{'}$ containing more weight, then
$W_{S}$ is not an optimal solution of separated strategy on segment
$l$ . It reduces a contradiction. So $W_{S}\ge W_{C}/2$. It reduces
$OPT_{S2}\ge OPT_{2}/2$. Let $W_{2}$ is the sum of weight of PoIs
covered by algorithm $A$. We reduce $W_{2}\ge\beta*OPT_{S2}\ge\frac{\text{\ensuremath{\beta}}}{2}*OPT_{2}$.

We get the solution of algorithm $A$ is $W_{1}+W_{2}\ge\frac{\beta}{2}*(OPT_{1}+OPT_{2})$.

The theorem is proved. 
\end{IEEEproof}
\begin{thm}
Algorithm 3 is an $\frac{1}{2}$-approximation algorithm for MWSCP
problem. 
\end{thm}
\begin{IEEEproof}
Algorithm 3 is an optimal dynamic programming algorithm for MWSCP
problem using separated strategy. According to Theorem 9, Algorithm
3 is an $\frac{1}{2}$-approximation algorithm for MWSCP problem. 
\end{IEEEproof}

\subsection{Max-Weight sweep coverage on path for general cases}

In this subsection, we discuss MWSCP for general cases, i.e. when
sensors have arbitrary velocities. We use 2 methods to solve the general
case of MWSCP problem. One is to use rounding and dynamic programming
method, by which we can get a $\frac{1}{2\alpha}$-approximation algorithm
scheme, where integer $\alpha\ge2$. The other is to use randomized
method, by which we can get a randomized $\text{\ensuremath{\frac{1}{2}}(1-1/e)}$-approximation
algorithm. By conditional expectation method, we can get a deterministic
$\text{\ensuremath{\frac{1}{2}}(1-1/e)}$-approximation algorithm
by derandomizing the randomized algorithm.

\subsubsection{Rounding and dynamic programming method for MWSCP}

Given a set of sensors $S$ with velocities $V$, we round the velocities
$V$ to $V^{'}$, make it a new set of sensors input $S^{'}$. Then
we quote Algorithm 1 for sensors $S^{'}$ to cover the given PoIs.
The algorithm 2 is shown below. In the rounding step, let $v_{d}=max\left\{ min~d_{min~\alpha},v_{l}\right\} $,
where $v_{l}$ denote the lowest velocity except $v_{i}=0$ and $d_{min~\alpha}$
denote the minimum distance between every continuous $\alpha+1$ PoIs.
For the sensors in $S$ with velocities lower than $v_{d}$, the velocities
will be set to zero. And for the sensors with velocities no lower
than $v_{d}$, we set $v_{i}^{'}=v_{d}*\alpha^{\lfloor log_{\alpha}(v_{i}/v_{d})\rfloor}$
. That is,

\[
\begin{cases}
v_{i}^{'}=0 & (v_{i}<v_{d},v_{i}\in V)\\
v_{i}^{'}=v_{d}*2^{\lfloor log_{\alpha}(v_{i}/v_{d})\rfloor} & (v_{i}\ge v_{d},v_{i}\in V)
\end{cases}
\]

Then the sensors would be rounded to $K=\lceil log_{\alpha}(v_{s}/v_{d})\rceil+1$
groups, where $v_{s}=max\{v_{i}|1\leq i\leq M\}$.

% Algorithm
\begin{algorithm}[t]
\SetAlgoNoLine \KwIn{ The number of sensors $M$, the number of
PoIs $N$, sensors' velocities $V=\{v_{1},v_{2},...,v_{M}\}$, PoIs'
weight $W=\{\omega_{1},\omega_{2},...,\omega_{N}\}$, sweep period
$T$, PoIs' locations $X=\{x_{1},x_{2},...,x_{N}\}$.} \KwOut{The
maximum sum of weight of PoIs sweep covered.}

\For{ $i\leftarrow1$ to $N-2$ }{ $d_{i,i+2}=x_{i+2}-x_{i}$\;}

set $v_{d}=max\{min\{v_{i}|v_{i}\neq0,1\leq i\leq M\},min\{d_{i,i+2}|1\leq i\leq N-2\}\}$\;

set $v_{s}=max\{v_{i}|1\leq i\leq M\}$\;

set $m=\lceil log_{\alpha}(v_{s}/v_{d})\rceil$\;

\For{ $i\leftarrow1$ to $M$ }{

\eIf{ $v_{i}\geq v_{d}$ }{ // For the sensors with velocities
no less than $v_{d}$\; set $\sigma=\lfloor log_{\alpha}(v_{i}/v_{d})\rfloor$
\; set $v_{i}^{'}=\alpha^{\sigma}*v_{d}$ \;}{ //For the sensors
with velocities less than $v_{d}$\; $v<v_{d}$\; $v_{i}^{'}=0$\;
}}

set $V^{'}=\{v_{i}^{'}|1\leq i\leq M\}$, $K=m+1$\;

initial an entry array $Tr(i)$ for $1\leq i\leq K$ \;

apply algorithm 1 to new input ($M$ sensors, $N$ PoIs, sensors'
velocities $V^{'}$, PoIs' weight $W$, the number of kinds of velocities
$K$, sweep period $T$, PoIs' locations $X$, $Tr(i)$ for $1\leq i\leq K$)\;

place the sensors with velocities $v_{i}$ to the locations we get
from algorithm 1 for the sensors with velocity $v_{i}^{'}$ \;

\Return $A_{1}(V^{'})$\caption{MWSCP-General-Cases}
\label{alg:two} 
\end{algorithm}

Now we prove Algorithm 2 is a $\frac{1}{2\alpha}$-approximation algorithm
scheme. 
\begin{thm}
Algorithm 2 is a $\frac{1}{2\alpha}$-approximation algorithm. 
\end{thm}
\begin{IEEEproof}
Given a set of sensors $S$ with velocities $V$, we can get a new
set of sensors $S^{'}$ with rounded velocities $V^{'}$ like algorithm
2. We assume $A_{1}(S)$ is the optimal separated strategy for $S$.
In $A_{1}(S)$ for every sensors $s_{i}\in S$, there is a covering
region $R(s_{i})$. For every sensor $s_{i}$ with velocity $v_{i}$
not lower than $v_{d}$, we set $v_{i}^{'}=v_{d}*\alpha^{\lfloor log_{\alpha}(v_{i}/v_{d})\rfloor}$
, then $v_{i}^{'}\le v_{i}\le\alpha v_{i}^{'}$. The region $R(s_{i})$
can be covered by $\alpha$ copies of the sensor $s_{i}^{'}$ with
velocity $v_{i}^{'}$ for each $s_{i}\in S$. According to Pigeonhole
principle, there must be a covering region $R^{'}(s_{i}^{'})$ covered
by one sensor with velocity $v_{i}^{'}$ having the sum of weight
no less than $w(s_{i})/\alpha$, where $\omega(s_{i})$ is the sum
of weight of PoIs on range $R(s_{i})$. For each sensor $s_{i}\in S$
with velocity lower than $v_{d}$, note that the covering region $R(s_{i})$
can not cover more than $\alpha$ PoIs, so $R(s_{i})$ can be also
covered by $\alpha$ copies of sensor with velocity 0. If we design
a new algorithm, let sensor $s_{i}^{'}\in S^{'}$ cover the the covering
region $R^{'}(s_{i}^{'})$ with no less than $w(s_{i})/\alpha$ weight
when $v_{i}\ge v_{d}$, otherwise, cover the PoI with most weight
in $R(s_{i})$ when $v_{i}<v_{d}$, then we can get a separated deployment
for $S^{'}$, which can obtain the sum of weight no smaller than $A_{1}(S)/\alpha$.

Because $A_{1}(S^{'})$ is an optimal separated strategy for $S^{'}$,
$A_{1}(S^{'})$ can at least get the same weight as the deployment
above.

According to theorem 9, That Algorithm 4 is a $\frac{1}{2\alpha}$-approximation
algorithm is proved. 
\end{IEEEproof}
Note that $\alpha$ is an integer and $\alpha\ge2$. Some counterexamples
show that $\alpha$ can not be fractional number. So we can get the
best performance guarantee is $\frac{1}{4}$-approximation by Algorithm
2. Algorithm 2 quote algorithm 1, so the time complexity is $O(N*(M/K)^{K})$
where $K=\lceil log_{\alpha}(v_{s}/v_{d})\rceil+1$. $\alpha$ is
a tradeoff between time complexity and performance guarantee. In most
cases, the maximum velocity is constant times of minimum nonzero velocity,
so $K$ is a constant to make the time complexity of Algorithm 2 is
acceptable.

.

\subsubsection{Randomized rounding algorithm for MWSCP}

Rounding and dynamic programming method gives a $\frac{1}{2\alpha}$-approximation
algorithm, where integer $\alpha\ge2$. Its best performance guarantee
is $\frac{1}{4}$-approximation. In this subsection, we use randomized
algorithm to analyze MWSCP problem and get an approximation algorithm
with better performance guarantee, $\frac{1}{2}(1-1/e)\approx0.31606$.
Besides, no matter the difference between the maximum velocity and
the minimum velocity, this algorithm takes polynomial time complexity.

Actually, we still analyze MWSCP problem using separated strategy
to get a $(1-1/e)$-approximation algorithm at first, and then get
a $\frac{1}{2}(1-1/e)$-approximation algorithm for MWSCP problem
without eliminating the cooperated strategy. We use the following
integer programming formulation to denote MWSCP using separated strategy
problem, where the variable $z_{l}$ indicates whether the $l^{th}$
PoI is covered, the variable $x_{ij}$ indicates whether sensor i
sweep cover from $j^{th}$ PoI, and $S_{ij}$ indicates the set of
PoIs covered if sensor i sweep covered from $j^{th}$ PoI. $x_{ij}=1$
means the PoIs in $S_{ij}$ are sweep covered. Then the constraint
(4) must hold for each PoI $z_{l}$ since at least one set $S_{ij}$
containing $z_{l}$ is included in solution when $z_{l}=1$ and no
set containing $z_{l}$ is included otherwise. The constraint (5)
means one sensor can be only used once.

\begin{eqnarray}
 & max\sum_{l=1}^{n}\omega_{l}z_{l}\nonumber \\
s.t. & \sum_{(i,j):l\in S_{ij}}x_{ij}\ge z_{l} & \forall l\in N\\
 & \sum_{j}x_{ij}\le1 & \forall i\in M\\
 & x_{ij}\in\{0,1\} & \forall i\in M,j\in N\nonumber \\
 & z_{l}\in\{0,1\} & \forall l\in N\nonumber 
\end{eqnarray}

Then we get the following linear programming relaxation from the integer
program above by replacing the constraints $x_{ij}\in\{0,1\}$ and
$z_{l}\in\{0,1\}$ with $0\le x_{ij}\le1$ and $z_{l}\le1$ .

\begin{eqnarray}
 & max\sum_{l=1}^{n}\omega_{l}z_{l}\nonumber \\
s.t. & \sum_{(i,j):l\in S_{ij}}x_{ij}\ge z_{l} & \forall l\in N\\
 & \sum_{j}x_{ij}\le1 & \forall i\in M\nonumber \\
 & 0\le x_{ij}\le1 & \forall i\in M,j\in N\nonumber \\
 & z_{l}\le1 & \forall l\in N\nonumber 
\end{eqnarray}

Let $(x^{*},z^{*})$ is the optimal solution to the linear program.
We use randomized rounding method to analyze MWSCP problem by making
sensor i to sweep cover from $j^{th}$ PoI with probability $x_{ij}^{*}$
independently. Then we can get a $(1-1/e)$-approximation randomized
algorithm for MWSCP using separated strategy problem.

% Algorithm
\begin{algorithm}[t]
\SetAlgoNoLine \KwIn{ The number of sensors $M$, the number of
PoIs $N$, sensors' velocities $V=\{v_{1},v_{2},...,v_{M}\}$, PoIs'
weight $W=\{\omega_{1},\omega_{2},...,\omega_{N}\}$, sweep period
$T$, PoIs' locations $X=\{x_{1},x_{2},...,x_{N}\}$.}\KwOut{The
maximum sum of weight of PoIs covered }

compute an optimal solution $(y^{*},z^{*})$ to the linear program
relaxation.

\For{ $i\leftarrow1$ to $M$ }{ make sensor $i$ cover from $j^{th}$
PoI independently at random with probability $x_{ij}^{*}$ .\; }

\Return 

\caption{Randomized rounding algorithm for MWSCP}
\label{alg:three} 
\end{algorithm}

\begin{thm}
Algorithm 3 is a randomized $\text{\ensuremath{\frac{1}{2}}(1-1/e)}$-approximation
algorithm for MWSCP. 
\end{thm}
\begin{proof} The fractional value $x_{ij}^{*}$ is interpreted as
the probability that $S_{ij}$ is chosen. The proof is similar to
the proof of theorem 5.10 in \cite{williamson2011design}. Then we
see the probability that $z_{l}$ is not covered is

\begin{eqnarray*}
Pr[z_{l}~is~not~coverd] & = & \prod_{(i,j):l\in S_{ij}}(1-x_{ij}^{*})\\
 & \le & \lbrack{\frac{1}{n_{l}}\sum_{(i,j):l\in S_{ij}}(1-x_{ij}^{*})}\rbrack^{n_{l}}\\
 & \le & (1-\frac{\sum x_{ij}^{*}}{n_{l}})^{n_{l}}\\
 & \le & (1-\frac{z_{l}^{*}}{n_{l}})^{n_{l}}
\end{eqnarray*}

Where $n_{l}$ indicates the number of sets in which $z_{l}$ is included
and the last inequality follows from the constraint (6).

When $n_{l}\ge1$, $f(z_{l}^{*})=1-(1-\frac{z_{l}^{*}}{n_{l}})^{n_{l}}$
is concave. So the probability that $z_{l}$ is covered is 
\begin{eqnarray*}
Pr[z_{l}~is~covered] & \ge & 1-(1-\frac{z_{l}^{*}}{n_{l}})^{n_{l}}\\
 & \ge & \lbrack{1-(1-\frac{1}{n_{l}})^{n_{l}}}\rbrack z_{l}^{*}
\end{eqnarray*}

Let W be a random variable that is equal to the total weight of the
covered PoIs. let $Y_{l}$ is a random variable such that $Y_{l}$
is 1 if $z_{l}$ is covered and 0 otherwise. We know that 
\begin{eqnarray*}
E[W] & = & \sum_{l=1}^{n}\omega_{l}E[Y_{l}]\\
 & = & \sum_{l=1}^{n}\omega_{l}Pr(z_{l}~is~covered)\\
 & \ge & \sum_{l=1}^{n}\omega_{l}z_{l}^{*}[1-(1-\frac{1}{n_{l}})^{n_{l}}]\\
 & \ge & min_{k\ge1}[1-(1-\frac{1}{k})^{k}]\sum_{l=1}^{n}\omega_{l}z_{l}^{*}\\
 & \ge & (1-\frac{1}{e})OPT
\end{eqnarray*}
Now we prove that we can get randomized $(1-1/e)$-approximation algorithm
for MWSCP using separated strategy problem by randomized rounding
technique. According to theorem 9, Algorithm 4 is $\frac{1}{2}(1-1/e)$-approximation
algorithm for MWSCP. \end{proof}

By the method of conditional expectations, we can obtain a deterministic
algorithm, Algorithm 4, that has the same performance guarantee as
the randomized one. So the algorithm 4 is a deterministic $\frac{1}{2}(1-1/e)$approximation
algorithm. Since it takes polynomial time of input to solve the linear
programming formulation, denoted as $p(N)$ and $M*N$ linear programming
formulations are needed to be solved in the derandomized algorithm.
The algorithm 4 is a polynomial algorithm and its time complexity
is $O(M*N*p(N))$.

% Algorithm
\begin{algorithm}[t]
\SetAlgoNoLine \KwIn{ The number of sensors $M$, the number of
PoIs $N$, sensors' velocities $V=\{v_{1},v_{2},...,v_{M}\}$, PoIs'
weight $W=\{\omega_{1},\omega_{2},...,\omega_{N}\}$, sweep period
$T$, PoIs' locations $X=\{x_{1},x_{2},...,x_{N}\}$.} \KwOut{The
maximum sum of weight of PoIs covered. }

\For{ $i\leftarrow1$ to $M$ }{ set $j_{i}=argmax_{j\in N}E(W|x_{1j_{1}}=1,...,x_{(i-1)j_{(i-1)}}=1,x_{ij}=1)$
.\; }

\For{ $i\leftarrow1$ to $M$ }{ make sensor $i$ to cover from
$j_{i}^{th}$ PoI. \; }

$St=\cup_{i}S_{ij_{i}}$\;

$\omega t=\sum_{z_{l}\in St}\omega_{l}$\;

\Return $\omega t$\caption{Deterministic version for MWSCP}
\label{alg:four} 
\end{algorithm}

Note that in the algorithm 4,

\[
\begin{aligned} & E[W|x_{1j_{1}}=1,...,x_{(i-1)j_{(i-1)}}=1,x_{ij}=1]\\
= & \sum_{l\in N}\omega_{l}Pr[z_{l}|x_{1j_{1}}=1,...,x_{(i-1)j_{(i-1)}}=1,x_{ij}=1]
\end{aligned}
\]

where $x_{ij_{i}}=1$ means that $j_{i}$ is the final optimal starting
location of sensor $i$.

Randomized method usually gets better approximation ratio than $\frac{1}{2}(1-1/e)$
in practice . 

\section{Conclusion}

There are many applications of sweep coverage on path, such as forest
patrol and intruder detection. In this paper, we first study the problem
of point sweep coverage on path and prove it is NP-Complete. We also
define its variant, the Max-Weight sweep coverage on path problem,
which are NP-Hard. For MWSCP, we propose an optimal algorithm for
the case that sensors have uniform velocity and a $\frac{1}{2}$-approximation
algorithm for sensors with constant number of velocities. For the
general case of MWSCP, we propose two algorithms. One is a $\frac{1}{2\alpha}$-approximation
algorithm family scheme, where integer $\alpha\ge2$. The other is
a $\frac{1}{2}(1-1/e)$-approximation algorithm. It is interesting
to study Max-Weight sweep coverage problem in other kinds of graphs,
e.g., trees, Eulerian graphs and so on. That is what we are about
to study.

\section*{Acknowledgment}

This work is supported by The 985 Project Funding of Sun Yat-sen University,
Australian Research Council Discovery Projects Funding DP150104871.

\bibliographystyle{plain}
\bibliography{Sweep-coverage}

\begin{thebibliography}{10}

\bibitem{chen2016efficient}
Zhiyin Chen, Xudong Zhu, Xiaofeng Gao, Fan Wu, Jian Gu, and Guihai Chen.
\newblock Efficient scheduling strategies for mobile sensors in sweep coverage
  problem.
\newblock In {\em Sensing, Communication, and Networking (SECON), 2016 13th
  Annual IEEE International Conference on}, pages 1--4. IEEE, 2016.

\bibitem{cheng2008sweep}
Weifang Cheng, Mo~Li, Kebin Liu, Yunhao Liu, Xiangyang Li, and Xiangke Liao.
\newblock Sweep coverage with mobile sensors.
\newblock In {\em Parallel and Distributed Processing, 2008. IPDPS 2008. IEEE
  International Symposium on}, pages 1--9. IEEE, 2008.

\bibitem{collins2013optimal}
Andrew Collins, Jurek Czyzowicz, Leszek Gasieniec, Adrian Kosowski, Evangelos
  Kranakis, Danny Krizanc, Russell Martin, and Oscar Morales~Ponce.
\newblock Optimal patrolling of fragmented boundaries.
\newblock In {\em Proceedings of the twenty-fifth annual ACM symposium on
  Parallelism in algorithms and architectures}, pages 241--250. ACM, 2013.

\bibitem{czyzowicz2011boundary}
Jurek Czyzowicz, Leszek G{\c{a}}sieniec, Adrian Kosowski, and Evangelos
  Kranakis.
\newblock Boundary patrolling by mobile agents with distinct maximal speeds.
\newblock In {\em European Symposium on Algorithms}, pages 701--712. Springer,
  2011.

\bibitem{du2010sweep}
Junzhao Du, Yawei Li, Hui Liu, and Kewei Sha.
\newblock On sweep coverage with minimum mobile sensors.
\newblock In {\em Parallel and Distributed Systems (ICPADS), 2010 IEEE 16th
  International Conference on}, pages 283--290. IEEE, 2010.

\bibitem{dumitrescu2014fence}
Adrian Dumitrescu, Anirban Ghosh, and Csaba~D T{\'o}th.
\newblock On fence patrolling by mobile agents.
\newblock {\em arXiv preprint arXiv:1401.6070}, 2014.

\bibitem{gary1979computers}
Michael~R Gary and David~S Johnson.
\newblock Computers and intractability: A guide to the theory of
  np-completeness, 1979.

\bibitem{gorain2013point}
Barun Gorain and Partha~Sarathi Mandal.
\newblock Point and area sweep coverage in wireless sensor networks.
\newblock In {\em Modeling \& Optimization in Mobile, Ad Hoc \& Wireless
  Networks (WiOpt), 2013 11th International Symposium on}, pages 140--145.
  IEEE, 2013.

\bibitem{gorain2014approximation}
Barun Gorain and Partha~Sarathi Mandal.
\newblock Approximation algorithms for sweep coverage in wireless sensor
  networks.
\newblock {\em Journal of Parallel and Distributed Computing},
  74(8):2699--2707, 2014.

\bibitem{gorain2014line}
Barun Gorain and Partha~Sarathi Mandal.
\newblock Line sweep coverage in wireless sensor networks.
\newblock In {\em 2014 Sixth International Conference on Communication Systems
  and Networks (COMSNETS)}, pages 1--6. IEEE, 2014.

\bibitem{gorain2015approximation}
Barun Gorain and Partha~Sarathi Mandal.
\newblock Approximation algorithm for sweep coverage on graph.
\newblock {\em Information Processing Letters}, 115(9):712--718, 2015.

\bibitem{gorain2016solving}
Barun Gorain and Partha~Sarathi Mandal.
\newblock Solving energy issues for sweep coverage in wireless sensor networks.
\newblock {\em Discrete Applied Mathematics}, 2016.

\bibitem{kawamura2015fence}
Akitoshi Kawamura and Yusuke Kobayashi.
\newblock Fence patrolling by mobile agents with distinct speeds.
\newblock {\em Distributed Computing}, 28(2):147--154, 2015.

\bibitem{kawamura2015simple}
Akitoshi Kawamura and Makoto Soejima.
\newblock Simple strategies versus optimal schedules in multi-agent patrolling.
\newblock In {\em International Conference on Algorithms and Complexity}, pages
  261--273. Springer, 2015.

\bibitem{kong2014surface}
Linghe Kong, Mingchen Zhao, Xiao-Yang Liu, Jialiang Lu, Yunhuai Liu, Min-You
  Wu, and Wei Shu.
\newblock Surface coverage in sensor networks.
\newblock {\em IEEE Transactions on Parallel and Distributed Systems},
  25(1):234--243, 2014.

\bibitem{li2015minimizing}
Shuangjuan Li and Hong Shen.
\newblock Minimizing the maximum sensor movement for barrier coverage in the
  plane.
\newblock In {\em 2015 IEEE Conference on Computer Communications (INFOCOM)},
  pages 244--252. IEEE, 2015.

\bibitem{pasqualetti2012cooperative}
Fabio Pasqualetti, Antonio Franchi, and Francesco Bullo.
\newblock On cooperative patrolling: Optimal trajectories, complexity analysis,
  and approximation algorithms.
\newblock {\em IEEE Transactions on Robotics}, 28(3):592--606, 2012.

\bibitem{sengupta2013multi}
Soumyadip Sengupta, Swagatam Das, MD~Nasir, and Bijaya~K Panigrahi.
\newblock Multi-objective node deployment in wsns: In search of an optimal
  trade-off among coverage, lifetime, energy consumption, and connectivity.
\newblock {\em Engineering Applications of Artificial Intelligence},
  26(1):405--416, 2013.

\bibitem{yu2014connected}
Zuoming Yu, Jin Teng, Xiaole Bai, Dong Xuan, and Weijia Jia.
\newblock Connected coverage in wireless networks with directional antennas.
\newblock {\em ACM Transactions on Sensor Networks (TOSN)}, 10(3):51, 2014.

\end{thebibliography}

\end{document}